\documentclass[a4paper,12pt]{article}

\usepackage{enumerate}
\usepackage{setspace}
\usepackage{amssymb, amsfonts, amsmath, amsthm, eurosym, mathabx, graphicx, subfig, srcltx, calc, epsfig, color, multirow, mathrsfs}
\usepackage[multiple]{footmisc}
\usepackage{enumerate}
\usepackage{mathrsfs}
\usepackage{float}
\usepackage{multicol}
\usepackage{bbm}
\usepackage{enumerate}
\usepackage{graphicx}
\usepackage{float}
\newtheorem{theorem}{\sc Theorem}

\newtheorem{proposition}{\sc Proposition}
\newtheorem{lemma}{\sc Lemma}
\newtheorem{claim}{\sc Claim}

\newtheorem{mstheorem}{\sc Theorem}

\newcommand{\argmax}{\mathop{\rm argmax}}

\newcommand{\sgn}{\mathop{\rm sgn}}
\newtheorem{definition}{\sc Definition}
\newtheorem{example}{\sc Example}

\allowdisplaybreaks

\begin{document}
\title{A dual approach to nonparametric characterization\\ for random utility models}
\author{Nobuo Koida\thanks{Faculty of Policy Studies, Iwate Prefectural University: nobuo@iwate-pu.ac.jp}\,  and Koji Shirai\thanks{School of Economics, Kwansei Gakuin University: kshirai1985@kwansei.ac.jp\newline \emph{The first version: March 7, 2024}. For helpful discussions and comments, the authors are grateful to V. Aguiar, C. Hara, Y. Higashi, N. Takeoka, and C. Turansick. We also thank the audience at North American Summer Meeting of Econometric Society (NASMES 2024, Vanderbilt) for helpful discussions.}}
\date{\today}
\maketitle

\abstract{This paper develops a novel characterization for random utility models (RUM), which turns out to be a dual representation of the characterization by Kitamura and Stoye (2018, ECMA). For a given family of budgets  and its ``patch" representation \emph{\'{a} l{a}} Kitamura and Stoye, we construct a matrix $\Xi$ of which each row vector indicates the structure of possible revealed preference relations in each subfamily of budgets. Then, it is shown that a stochastic demand system on the patches of budget lines, say $\pi$, is consistent with a RUM, if and only if $\Xi\pi\geq \mathbbm{1}$, where the RHS is the vector of $1$'s. In addition to providing a concise quantifier-free characterization, especially when $\pi$ is inconsistent with RUMs, the vector $\Xi\pi$ also contains information concerning (1) sub-families of budgets in which cyclical choices must occur with positive probabilities, and (2) the maximal possible weight on rational choice patterns in a population. The notion of Chv\'{a}tal rank of polytopes and the duality theorem in linear programming play key roles to obtain these results.}\vspace{0.2in}

\vspace{-0.1in}

\clearpage

\onehalfspacing
\section{Introduction}

In the literature of random utility models (RUM), Kitamura and Stoye (2018) (henceforth, KS) has established a powerful and tractable analytical tool for nonparametric demand analysis. Based on several fundamental results in revealed preference theory by Afriat (1967), Varian (1982) and McFadden and Richter (1990), they provide an insightful characterization for stochastic demand systems consistent with a RUM, as well as a statistical procedure for testing it based on empirical data. Afterward, in the literature, their approach has been further developed and turned out to be useful in various models. For example, Smeulders, Cherchye and De Rock (2021) develops an efficient computation technique for implementing the analysis in KS, while the papers including Aguiar, Gautheir, Kashaev and Plavala (2023), Deb, Kitamura, Quah and Stoye (2023) and Lazzati, Quah and Shirai (2024) show the applicability of Kitamura and Stoye's approach beyond the classical consumer theory, with some methodological contributions being also made in each of them.\footnote{Aguiar et al. (2023) deals with a dynamic consumption model, while Deb et al. (2023) works on the model of price preferences. Lazzati et al. (2024) applies KS' approach to game theoretical setting. }

Subsequent to these works, in this paper, we evolve the approach of KS from a theoretical perspective, especially in the framework of  consumer theory. To be specific, this paper provides a novel necessary and sufficient condition for a stochastic demand system to be consistent with RUMs. Our characterization turns out to be a dual representation of that by KS, which has several attractive features. In terms of a formal aspect, we obtain a concise and easy-to-interpret quantifier-free condition, rather than solvability/satisfiability type conditions.\footnote{In the framework of abstract choice theory, by Block and Marschack (1960), the famous (quantifier-free) characterization, so called \emph{Block-Marschack polynominals}, is known. Their condition is an extension of the monotonicity of choice frequencies with respect to the set inclusion relation on choice sets, rather than the structure of revealed preference relation. See Kono, Saito and Sandroni (2023) and Turansick (2024) for recent development in Block-Marschack type argument.} As a benefit from our representation, when a given stochastic demand system is inconsistent with RUMs, it simultaneously detects ``where" the consistency breaks down. Furthermore, using the duality with the characterization by KS, it is also shown that our condition identifies ``to what extent" a given stochastic demand system is (in)consistent with RUMs.

From a technical viewpoint, KS characterizes the set of RUM-consistent stochastic demand systems as a polytope of which vertices are deterministic choice patterns obeying the {strong axiom of revealed preference (SARP)}. Since the polytope is described by using the set of vertices, the characterization in KS is called a ${\cal V}$-representation. On the other hand, in this paper, we characterize the same polytope in terms of the set of hyperplanes generating it, which is often referred to as an ${\cal H}$-representation. Once a ${\cal V}$-representation is obtained, by Minkowski-Weyl duality, it is \emph{theoretically} straightforward that there exists an ${\cal H}$-representation. However, as KS pointed out, this connection is purely theoretical and it is quite nontrivial to explicitly obtain it from the ${\cal V}$-representation by KS, let alone its economic implication. Indeed, the potential benefits from having an ${\cal H}$-representation are recognized in the literature, but it has not obtained beyond some specific numerical examples (see, for example, Aguiar et al. (2023) as well as KS).

Instead of starting from the ${\cal V}$-representation by KS, we directly construct a set of hyperplanes that captures observable restrictions from utility maximizing behavior, and rediscover the polytope of RUM-consistent demand systems by using it. More precisely, a keystone of our approach is constructing a matrix that captures the structure of revealed preference relations across budgets. This matrix is formed so that it provides a quantifier-free characterization for a (deterministic) choice pattern to obey SARP; that is, a set of hyperplanes determining SARP-consistent consumption patterns is specified. Then, we show that the same set of hyperplanes in fact generates the polytope corresponding to the set of stochastic demand systems consistent with RUMs. This means that the set of vertices of that polytope coincides with the set of SARP-consistent choice patterns. Since deterministic choice patterns are represented as binary vectors in our setting, the above property in turn corresponds to the \emph{integrality} of polytopes. To prove it, we employ the notion of \emph{{Chv\'{a}tal rank}}, which is a well-known concept in integer programming for, loosely speaking, measuring the degree of non-integrality of a given polytope. We prove that Chv\'{a}tal rank of the above polytope is equal to $0$, which is equivalent to the integrality of that polytope.\footnote{To the best of the authors' knowledge, this is the first attempt to apply the notion of Chv\'{a}tal rank in the literature of revealed preference theory, despite many successful applications of integer programming there. } 

It should be also noted that our approach can be interpreted as an extension of that in Hoderlein and Stoye (2014), which characterizes the set of stochastic demand systems consistent with the {weak axiom of revealed preference (WARP)}. Their characterization is also given as a quantifier-free linear condition, and in fact it can be captured as a \emph{subsystem} of our necessary and sufficient condition for RUMs. In this aspect, the current paper bridges between the condition for WARP-consistent stochastic choices by Hoderlein and Stoye (2014) and that for SARP-consistent stochastic choices by KS. Despite characterizing closely related models, the connection between these conditions has not necessarily been clear.

The rest of this paper is arranged as follows. In Section 2, we introduce the basic setting in the paper, and briefly explain the characterization for RUMs by KS. Then, in Section 3.1, our alternative characterization is established. There, it is also shown that, when a given demand system is inconsistent with RUMs, our necessary and sufficient condition specifies subfamilies of budgets where cyclical choices are crucial. We raise several numerical examples in Section 3.2, and proceed to the proof of the characterization theorem in Section 3.3. In Section 4.1, we formally show the duality between our characterization and that by KS, which immediately uncovers the connection between our characterization and the identification of the maximal fraction of rational choices. Some numerical examples are given in Section 4.2, and the proofs for the results in Section 4.1 are contained in Section 4.3. Lastly, in Section 5, we conclude the paper, with referring to some possible directions of future researches.

\section{Rationalizability of random consumption}

Throughout this paper, we follow the framework of KS, which is based on the classical consumer model. Suppose that there are $n$ $(\geq 2)$ commodities for which nonnegative consumption levels are allowed under positive price vectors. An increasing utility function is denoted by $u:\mathbb{R}^n_+\rightarrow \mathbb{R}$, and a \emph{random utility model (RUM)} is defined as a distribution of these utility functions, which is in turn denoted by $\Phi$. 


Our first objective is to characterize the observable restrictions from RUMs on choice behavior on finitely many budgets. Suppose that there are $J<\infty$ fixed budgets $B_j=\{y\in \mathbb{R}^n_+: p_j\cdot y=1\}$, where $p_j=(p_{j1},...,p_{jn})\in \mathbb{R}^n_{++}$ is a positive price vector for $j=1,2,...,J$. We also assume that cross-sectional distribution of demand corresponding to these budgets are observed; that is, we work with a population distribution, rather than any kind of empirical data. Denoting a distribution of demand on each ${\cal B}_j$ by $P_j(S)$ for $S\subset {\cal B}_j$, we call a profile of them, say, $P=(P_1,P_2,...P_J)$ as a \emph{stochastic demand system}. The consistency of it with RUMs is defined as follows.

\begin{definition}
A stochastic demand system $P=(P_1,P_2,...,P_J)$ is rationalizable, if there exists a RUM $\Phi$ such that\begin{align}\label{def1}
P_j(S)=\int {\bf 1}\left(\argmax_{y\in B_j}u(y)\in S \right)d\Phi(u),\mbox{ for every }S\subset {\cal B}_j\mbox{ and }j=1,2,...,J.
\end{align}
\end{definition}\vspace{0.05in}

KS established a simple, but insightful geometric approach for characterizing the above defined rationalizability. A key idea for that is making \emph{patches} of budget lines, using a kind of equivalent classes with respect to the direct revealed preference relations. To be specific, each budget set $B_j$ is divided into patches $\left(B_{j1},B_{j2},...,B_{jI_j}\right)$ defined such that\begin{align}
\sgn(p_{j'}\cdot y'-1) &=\sgn(p_{j'}\cdot y''-1)\mbox{ for all }j'\neq  j\\
&\iff y',y''\in B_j\mbox{ are contained in the same patch }B_{ji}.\notag
\end{align}
As seen from the definition, consumption vectors obtained from the same set of patches would derive the same direct revealed preference relations. In what follows, we always assume that \emph{any intersection of budget lines is not chosen with any positive probability}, which is ensured if distributions of demands are continuous. By this assumption, as argued in KS, it suffices to consider patches belonging to a single budget set. Figure \ref{fig_kspatch} visualizes the construction of patches on budgets, where each patch excludes the intersection of two budget lines. In what follows, let $I=\sum^J_{j=1}I_j$; that is, $I$ is the total number of patches.

\begin{figure}[tb]
\centering
\includegraphics[keepaspectratio, scale=0.4]{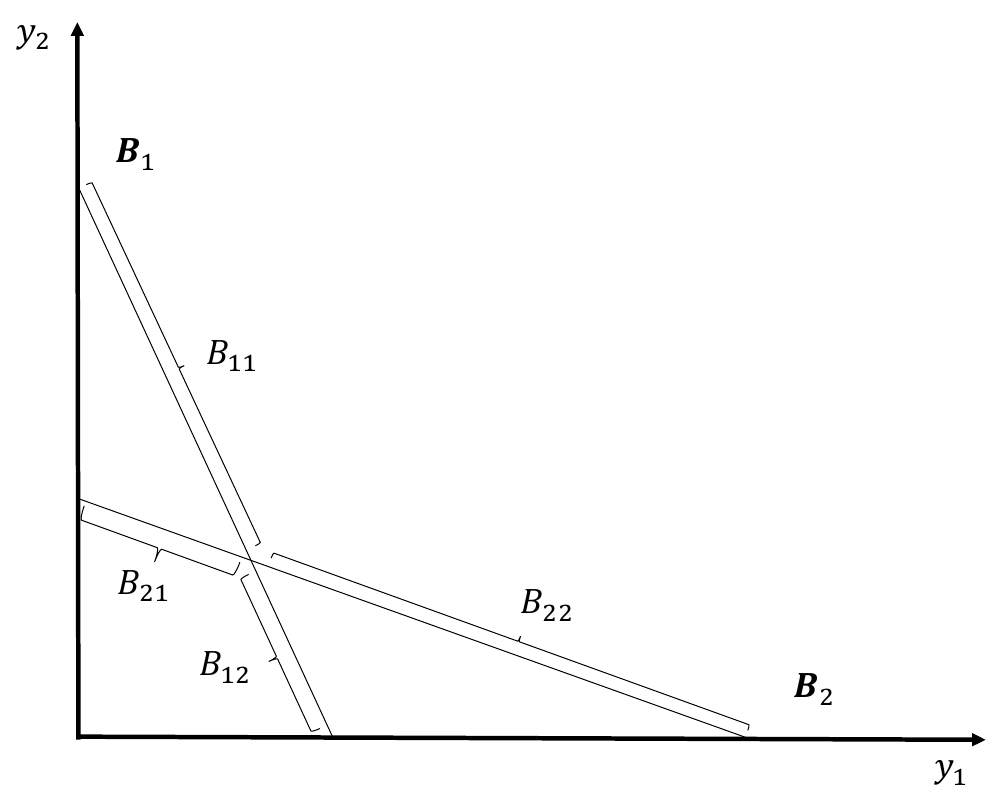}
\caption{Patches on budget lines.}
\label{fig_kspatch}
\end{figure}

Using the notion of patches, we obtain the \emph{vector representation} of a stochastic demand system as \begin{align}
\pi = (\pi_1,\pi_2,...,\pi_J)=(\pi_{11},\pi_{12},...,\pi_{1I_1};...,;\pi_{J1},...,\pi_{JI_J}),
\end{align}
where each $\pi_j:=\left(\pi_{j1},\pi_{j2},...,\pi_{jI_j}\right)$ is a probability vector on $\left(B_{j1},B_{j2},...,B_{jI_j}\right)$, and hence, each $\pi_{ji}$ stands for a probability mass put on the patch $B_{ji}$.  (Note that, in the right most side, the semicolons indicate the separations of budgets, and we use this notation throughout this paper.) In fact, the rationalizability of a stochastic demand system can be tested through the property of its vector representation. Specifically, a stochastic demand system $P$ is rationalizable, if and only if the corresponding $\pi$ is represented as a convex combination of rational non-stochastic choice patterns explained below. This fact is extensively used also in our approach.

A deterministic choice pattern, referred to as a \emph{behavioral types}, is formally defined as \begin{align}
a=(a_{11},a_{12},...,a_{1I_1};...;a_{J1},a_{J2},...,a_{JI_J}),
\end{align}
where each $a_j:=\left(a_{j1},a_{j2},...,a_{jI_j}\right)$ is a binary vector with $\sum^{I_j}_{i=1} a_{ji}=1$. That is, each $a_j$ specifies one and only one patch $B_{ji}$, which can be interpreted as a choice from $B_j$. Abusing notation, let $a(B_j)=\left\{B_{ji}:a_{ji}=1\right\}$ and define the \emph{direct revealed preference} $\succ^R$ such that\begin{align}\label{drp}
a(B_{j''})\succ^R a(B_{j'}),\mbox{ if } y''\in a(B_{j''}),y'\in a(B_{j'})\Longrightarrow p_{j''}\cdot y'-1<0.
\end{align}
Since we do not consider the intersections of budget lines, it suffices to consider the case of a strict inequality. When it holds that for some ${\cal J}:=\{j_1,j_2,...,j_l\}\subset \{1,2,...,J\}$,\begin{align}\label{cycle_wrt}
a(B_{j_1})\succ^R a(B_{j_2})\succ^R\cdots \succ^R a(B_{j_l})\succ^R a(B_{j_1}),
\end{align} 
we say that the behavioral type has a \emph{revealed preference cycle}. A behavioral type is \emph{rationalizable}, if it obeys the \emph{strong axiom of revealed preference (SARP)} in the sense that it does not have any revealed preference cycle.\footnote{In general, the rationalizability of (deterministic) consumer choices is characterized as the generalized axiom of revealed preference (GARP), which is slightly weaker than SARP (see, Afriat (1967) and Varian (1982)). Nevertheless, these two notions coincide under our assumption excluding the demand at the intersection of budgets. }

Let ${\cal A}$ be the set of all behavioral types, and ${\cal A}^*\subset {\cal A}$ be the set of all rationalizable behavioral types. Similarly, let $A$ be the matrix of which the set of column vectors is equal to ${\cal A}$, and $A^*$ be the matrix of which the set of column vectors is equal to ${\cal A}^*$. Then, Kitamura and Stoye's characterization theorem is given as follows.

\setcounter{theorem}{-1}
\begin{theorem}\label{ks}
A stochastic demand system $P=(P_1,P_2,...,P_J)$ is rationalizable, if and only if its vector representation $\pi=(\pi_1,\pi_2,...,\pi_J)$ is represented as $\pi=A^*\tau^*$ for some $\tau^*\geq 0$. 
\end{theorem}

Thus, the above theorem characterizes the rationalizability through the existence of a nonnegative vector that solves the system of linear equation $\pi=A^*\tau^*$. In fact, as shown by KS, $\tau^*$ automatically satisfies the adding-up condition, and hence $\pi$ is represented as a weighted sum of rationalizable behavioral types.  Theorem \ref{ks} also implies that the set of rationalizable stochastic demand systems is captured as a polytope, since it is characterized as a convex hull of rationalizable behavioral types. In general, the representation of a polytope in terms of its vertices (as in Theorem \ref{ks}) is referred to as a ${\cal V}$-represetantion. 

On the other hand, it is well known that a polytope can be also represented as an intersection of finitely many half spaces of hyperplanes, which is referred to as an ${\cal H}$-representation. Once one of these representations is obtained, Minkowski-Weyl duality immediately implies the existence of the other representation. However, the existence here is purely theoretical and, in general, it is quite non-trivial to explicitly construct it. (See, for example, Ziegler (2007) for the detail.) Despite that, some specific structure of consumer problem allows us to establish an explicit ${\cal H}$-representation of Theorem \ref{ks}, which turns out to have several attractive economic implications. Construction of an alternative characterization is the goal of Section 3, and the duality with Theorem \ref{ks} is explored in Section 4.

\vspace{-0.1in}
\section{An alternative characterization}
\vspace{-0.1in}
\subsection{Characterization by hyperplanes}

For investigating the rationalizability of stochastic demand systems, by Theorem \ref{ks}, it suffices to look at its vector representation. Hence, in the rest of this paper, we always deal with a stochastic demand system by its vector representation $\pi$, and we simply say that $\pi$ is (not) rationalizable when the underlying $P$ is (not) rationalizable.

A key idea for constructing the alternative characterization for RUMs is capturing the structure of possible revealed preference relation across patches. In particular, the following notion plays crucial roles in our analysis. Let ${\mathscr J}=\{{\cal J}\subset \{1,2,...,J\}: |{\cal J}|\geq 2\}$. For each ${\cal J}\in {\mathscr J}$, we say that a patch $B_{ji}$ is \emph{undominated in} ${\cal J}$, if; \begin{align}\label{max}
j\in {\cal J},\mbox{ and } p_{j'}\cdot y-1>0\mbox{ for all }y\in B_{ji}\mbox{ and }j'\in {\cal J}\setminus \{j\}.
\end{align}
That is, for each subfamily of budgets ${\cal J}\in \mathscr{J}$, a patch $B_{ji}$ is undominated, if it is not dominated by any other patches turning up in ${\cal J}$ with respect to the direct revealed preference relation defined in (\ref{drp}).\footnote{Put otherwise, a patch is undominated, if it is \emph{maximal} in ${\cal J}$ with respect to the direct revealed preference relation. However, we use the notion of maximality/minimality with respect to set inclusion in other parts of the paper, so we simply use the term ``undominated" instead of ``maximal.".} (Otherwise, a patch is said to be \emph{dominated} in ${\cal J}$.) Thus, if $a(B_j)=B_{ji}$ and it is undominated in ${\cal J}$, then there is no budget $B_{j'}$ for which $a(B_{j'})\succ^R a(B_j)$ with $j'\in \cal J$. Note that, as a basic property of undominatedness of patches, if $j\in {\cal J}'\subset {\cal J}''$ and a patch $B_{ji}$ is undominated in ${\cal J}''$, then it is also undominated in ${\cal J}'$. 

Using the notion of undominated patches, we define the vector that can ``detect" the existence of cyclical choices in each subfamily of budgets. Let for each ${\cal J}\in \mathscr{J}$, define $\xi^{\cal J}=(\xi^{\cal J}_{11},...,\xi^{\cal J}_{1I_1};...;\xi^{\cal J}_{J1},...,\xi^{\cal J}_{JI_J})\in \{0,1\}^I$ such that \begin{align}\label{def_xi}
\xi^{\cal J}_{ji}={\bf 1}\left( \mbox{$B_{ji}$ is undominated in ${\cal J}$}\right).
\end{align}
Once $\xi^{\cal J}$ are obtained as above for all ${\cal J}\in \mathscr{J}$, let $\Xi$ be the $(|{\mathscr J}|\times I)$-matrix such that each row vector corresponds to each $\xi^{\cal J}$. All our results are derived from the nature of the vector $\xi^{\cal J}$ and matrix $\Xi$. Hence, before checking the properties of them, we raise a simple example to clarify how they are constructed (as well as the notion of undominated patches). 

\begin{figure}[tbp]
\centering
\includegraphics[keepaspectratio, scale=0.6]{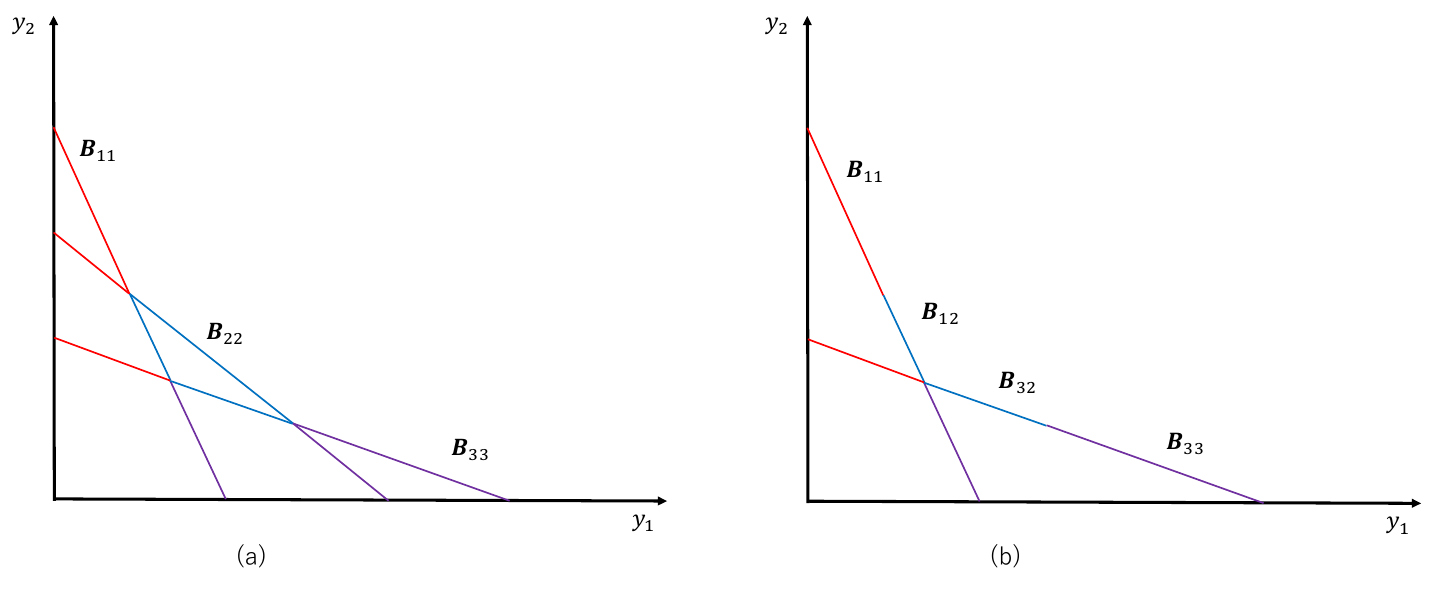}
\caption{Undominated patches in ${\cal J}=\{1,2,3\}$ and ${\cal J}=\{1,3\}$.}
\label{fig_maxpatch}
\end{figure}

\begin{example}\label{ex1}
Assume that $n=2$ and $J=3$, and consider the budget lines depicted in Figure \ref{fig_maxpatch}. In the figure, the red segments denote patches $B_{j1}$, the blue segments denote patches $B_{j2}$, and the purple segments denote patches $B_{j3}$ for $j=1,2,3$, respectively. Looking at the budget $B_1$, it is easy to confirm that the patch $B_{11}$ is undominated in $\{1,2,3\}$, which in turn implies that $B_{11}$ is undominated for the families of budgets $\{1,2\}$ and $\{1,3\}$. On the other hand, the patch $B_{12}$ is a dominated patch in $\{1,2,3\}$, since it is dominated by $B_{22}$ in terms of $\succ^R$. However, it \emph{is} undominated in $\{1,3\}$, since no patch in $B_3$ can dominate it. Repeating this argument, we obtain;
$\xi^{\{1, 2, 3\}}$ $=$ $(1, 0, 0; 0, 1, 0; 0, 0, 1)$,
$\xi^{\{1, 2\}}$ $=$ $(1, 0, 0; 0, 1, 1; 0, 0, 0)$,
$\xi^{\{2, 3\}}$ $=$ $(0, 0, 0; 1, 1, 0; 0, 0, 1)$, and
$\xi^{\{3, 1\}}$ $=$ $(1, 1, 0; 0, 0, 0; 0, 1, 1)$.
Accordingly, we obtain a $(|{\mathscr J}| \times I)$-matrix
\begin{align}\label{Xi_eg}
\Xi =
\begin{pmatrix} 
1 & 0 & 0 & 0 & 1 & 0 & 0 & 0 & 1\\
1 & 0 & 0 & 0 & 1 & 1 & 0 & 0 & 0\\
0 & 0 & 0 & 1 & 1 & 0 & 0 & 0 & 1\\
1 & 1 & 0 & 0 & 0 & 0 & 0 & 1 & 1
\end{pmatrix},
\end{align}
which will be repeatedly used in the examples in the rest of the paper.
\end{example} \vspace{0.05in}

To see the basic property concerning $\xi^{\cal J}$, notice that for $a\in {\cal A}$, $\xi^{\cal J}\cdot a=0$ implies that none of selected patch $a(B_j)$ is undominated in ${\cal J}$. This immediately implies that the direct revealed preference relation $\succ^R$ has to admit at least one cycle within ${\cal J}$, since there are only finitely many budgets. In this sense, each $\xi^{\cal J}$ can detect whether a behavioral type $a$ has revealed preference cycles within $\{B_j\}_{j\in {\cal J}}$, and hence, the matrix $\Xi$ provides an alternative representation of SARP. Given the importance of this claim in the paper, we summarize it as a lemma and provide a formal proof. See also Remark 1 below Lemma \ref{lem_sarp} for a more precise argument on the connection between $\xi^{\cal J}$ and revealed preference cycles.

\begin{lemma}\label{lem_sarp}
A behavioral type $a$ is rationalizable, if and only if it satisfies $\Xi a\geq \mathbbm{1}$, where $\mathbbm{1}$ is the $|{\mathscr J}|$-dimensional column vector consisting of $1$'s. 
\end{lemma}

\begin{proof}
Suppose that $a\in {\cal A}$ is not rationalizable. Then, the violation of SARP implies the existence of some ${\cal J}=\{j_1,j_2,...,j_l\}\in \mathscr{J}$ on which the choices made by $a$ forms an $\succ^R$-cycle such as\begin{align}\label{cycle1}
a(B_{j_1})\succ^R a(B_{j_2})\succ^R\cdots \succ^R a(B_{j_l})\succ^R a(B_{j_1}).
\end{align}
Obviously, none of patches turning up in the above cycle is undominated in ${\cal J}$, and hence $\xi^{\cal J}\cdot a=0$, which in turn implies that $\Xi a\geq \mathbbm{1}$ does not hold. Conversely, if a behavioral type $a\in {\cal A}$ does not have any $\succ^R$-cycle as in (\ref{cycle1}) for any ${\cal J}\in \mathscr{J}$, then it must choose at least one undominated patch in ${\cal J}$; that is, $\xi^{\cal J}\cdot a\geq 1$ must hold for every ${\cal J}\in \mathscr{J}$. This immediately shows that if $a$ is rationalizable, then $\Xi a\geq \mathbbm{1}$.
\end{proof}

\noindent
{\sc Remark 1.} Concerning the ``test" for the existence of cycle using $\xi^{\cal J}$, it should be noted that $\xi^{\cal J}\cdot a\geq 1$ does \emph{not} necessarily mean that $\succ^R$ is acyclic on $\{B_j\}_{j\in {\cal J}}$. For instance, in the situation of Example \ref{ex1}, if $a=(0,1,0;1,0,0;0,0,1)$, then $\xi^{\{1,2,3\}}\cdot a=1$. However, as $\xi^{\{1,2\}}\cdot a=0$ suggests, it contains a cycle $a(B_1)\succ^R a(B_2)\succ^R a(B_1)$. On the other hand, if $\xi^{{\cal J}'}\cdot a\geq 1$ for any ${\cal J}'\subset {\cal J}$, then it implies that $\succ^R$ is acyclic on $\{B_j\}_{j\in {\cal J}}$. Relatedly, if $\xi^{\cal J}\cdot a=0$ and there is no ${\cal J}'\subset {\cal J}$ for which $\xi^{{\cal J}'}\cdot a=0$, then it implies the existence of a revealed preference cycle involving all elements of $\{a(B_j)\}_{j\in {\cal J}}$; that is, letting ${\cal J}=\{j_1,j_2,...,j_l\}$, there is a cycle such as $a(B_{j_1})\succ^R a(B_{j_2})\succ^R\dots \succ^R a(B_{j_l})\succ^R a(B_{j_1}),$ possibly by adjusting the indices.\vspace{0.1in}   

While Lemma \ref{lem_sarp} ensures that the rationalizability of behavioral types can be tested by the system of inequalities (or equivalently, by the set of hyperplanes) generated by $\Xi$ and $\mathbbm{1}$, perhaps strikingly, it extends to the rationalizability of a stochastic demand system. 

\begin{theorem}\label{main}
A stochastic demand system $\pi$ is rationalizable, if and only if it satisfies $\Xi \pi \geq \mathbbm{1}$. \end{theorem}\vspace{0.1in}

This is an alternative representation of the ``test" for rationalizability of a given stochastic demand system, and hence, it is logically equivalent to Theorem \ref{ks}. On the other hand, the condition in Theorem \ref{main} is a quantifier-free characterization of random utility models. In addition, it characterizes the set of rationalizable demand systems as the intersection of half spaces of hyperplanes $\{\pi: \xi^{\cal J}\cdot \pi=1\}$ for ${\cal J}\in \mathscr{J}$. Thus, our characterization is an ${\cal H}$-representation of the polytope of those demand systems, of which the duality between the ${\cal V}$-representation in Theorem \ref{ks} is shown in the next section. Further mathematical argument concerning this theorem is postponed to Section 3.3, since it works as a good introduction to the formal proof stated there.

The condition $\Xi\pi\geq \mathbbm{1}$ itself requires that the sum of choice frequencies across undominated patches should not be smaller than $1$ for every subfamily of budgets. Intuitively, this implies that the total weight on cyclical choices is not too large, and hence, the choices in a population is explained by a distribution of rational choices. This intuition is reminiscent of a necessary and sufficient condition for the consistency with \emph{weak axiom of revealed preference (WARP)}, established by Hoderlein and Stoye (2014).\footnote{Note that WARP requires the asymmetry of the direct revealed preference relation, or the lack of choice reversal between any pair of budgets. In addition, to be precise, Hoderlein and Stoye (2014) derived the upper bound and the lower bound of WARP-consistent behavior in a population, as well as the statistical procedure for estimating them from real data.} 
Their condition essentially requires that for every pair of budgets, the sum of choice frequencies across WARP-violating combination of patches should not exceed $1$, which is clearly equivalent to requiring  $\xi^{\cal J}\cdot \pi\geq {1}$ for every ${\cal J}$ consisting of two budgets. For example, using the budget lines in Example \ref{ex1}, their condition requires that $(\pi_{12}+\pi_{13})+\pi_{21}\le 1$, $\pi_{23}+(\pi_{31}+\pi_{32})\le 1$, and $\pi_{13}+\pi_{31}\le 1$, which is equivalent to $\xi^{\{1,2\}}\cdot \pi\geq 1$, $\xi^{\{2,3\}}\cdot \pi\geq 1$ and $\xi^{\{1,3\}}\cdot \pi\geq 1$. Thus, our result can be interpreted as an extension of Hoderlein-Stoye approach to the case of RUMs, or stochastic choices consistent with SARP.  Relatedly, since WARP and SARP are equivalent in the two-commodity model, for checking the rationalizability in such a case, it suffices to consider $\xi^{\cal J}$ for ${\cal J}$ with $|{\cal J}|=2$ (\emph{i.e.} Hoderlein and Stoye's condition). However, as we will see in the next section, the value of $\xi^{\cal J}\cdot \pi$ for ${\cal J}\in \mathscr{J}$ with $|{\cal J}|\geq 3$ \emph{can} have some substantial information even in the two-commodity setting.

As a benefit of having the characterization in Theorem \ref{main}, when $\pi$ is not rationalizable, it allows us to obtain some information about ``where" the rationality breaks down. If $\pi$ is not rationalizable, then there exists some ${\cal J}\in \mathscr{J}$ such that $\xi^{\cal J}\cdot \pi<1$. Hence, in order to represent $\pi$ as a convex combination of behavioral types, it is inevitable to put positive weights on some behavioral types obeying $\xi^{\cal J}\cdot a=0$. By Lemma \ref{lem_sarp}, this in turn implies that revealed preference cycles within the budget family $\{B_j\}_{j\in {\cal J}}$ is crucial to explain $\pi$. In particular, given the fact stated in Remark 1, if ${\cal J}=\{j_1,j_2,...,j_l\}$ is a minimal element of $\mathscr{J}$ obeying $\xi^{{\cal J}}\cdot \pi<1$, then we have to put positive weights on some behavioral types containing a revealed preference cycle involving all elements of ${\cal J}$ (the one like (\ref{cycle1})). We summarize this as a proposition for future references.

\begin{proposition}\label{violation}
Suppose that $\Xi\pi\not\geq \mathbbm{1}$ holds for a given stochastic demand system $\pi$. Then, for every ${\cal J}\in \mathscr{J}$ with $\xi^{\cal J}\cdot\pi<1$ and $\tau\in \Delta(\cal A)$, \begin{align}
\pi=\sum_{a\in {\cal A}}\tau_aa\Longrightarrow \sum_{a: \xi^{\cal J}\cdot a<1}\tau_a>0.
\end{align}
In particular, if ${\cal J}$ is a minimal element of the set $\{{\cal J}'\in\mathscr{J}: \xi^{{\cal J}'}\cdot\pi<1\}$ (with respect to the set inclusion), then some revealed preference cycle consisting of indices in ${\cal J}$ in the sense of (\ref{cycle1}) must occur with a positive probability.
\end{proposition}

\vspace{-0.2in}
\subsection{Numerical examples}

Below, we raise three examples which respectively correspond to (i) a two-commodity and three-budget example where the stochastic demand system is rationalizable, (ii) a two-commodity and three-budget example where the stochastic demand is not rationalizable, and (iii) a three-commodity and three-budget example where the stochastic demand system is not rationalizable.  As explained above, the first two cases can also be dealt with by Hoderlein and Stoye's condition, but it would help how our characterization works in a simple setting. The third example satisfies the condition by Hoderlein and Stoye, but not the condition in Theorem \ref{main}. 

\begin{example} \label{ex2}
Consider the budgets depicted in Figure \ref{fig_maxpatch}, in which, as shown in Example \ref{ex1},  the matrix $\Xi$ is calculated as in (\ref{Xi_eg}). If a stochastic demand system is specified as $$\pi=\left(\frac{1}{3},\frac{1}{3},\frac{1}{3};\frac{1}{3},\frac{1}{3},\frac{1}{3};\frac{1}{3},\frac{1}{3},\frac{1}{3}\right)^t,$$ then it holds that $$\Xi \pi = \left(1,1,1,\frac{4}{3}\right)^t\geq \mathbbm{1}.$$
Hence $\pi$ is rationalizable by Theorem 1. For example, $\pi$ can be represented by the distribution $\tau$ 
that assigns probability of $1/3$ to each of behavioral types
$a^1=(1, 0, 0; 1, 0, 0; 1, 0, 0)^t$, $a^2=(0, 1, 0; 0, 1, 0; 0, 1, 0)^t$, and
$a^3=(0, 0, 1; 0, 0, 1; 0, 0, 1)^t$. Using Lemma \ref{lem_sarp}, it is straightforward to see that they are all rationalizable behavioral types, and hence, Theorem \ref{ks} also ensures the rationalizability of this stochastic demand system.
\end{example}

\begin{example}\label{ex3}
Again, consider the same budgets as in the preceding example, and let $$\pi=\left(\frac{1}{10},\frac{8}{10},\frac{1}{10};\frac{4}{10},\frac{2}{10},\frac{4}{10};\frac{1}{10},\frac{8}{10},\frac{1}{10}\right)^t.$$ In this case, $$\Xi\pi=\left(\frac{2}{5},\frac{7}{10},\frac{7}{10},\frac{9}{5}\right)^t\not\geq \mathbbm{1},$$ and hence, the stochastic demand system is not rationalizable. In addition, by Proposition \ref{violation}, if $\pi=\sum_{a\in {\cal A}}\tau_aa$ is satisfied for some $\tau\in \Delta({\cal A})$, then it must hold that $\sum_{a:\xi^{\{1,2,3\}}\cdot a=0}\tau_a>0$,  $\sum_{a:\xi^{\{1,2\}}\cdot a=0}\tau_a>0$, and  $\sum_{a:\xi^{\{2,3\}}\cdot a=0}\tau_a>0$ (recall the construction of $\Xi$ in Example \ref{ex1}). In particular, since $\{1,2\}$ and $\{2,3\}$ are minimal within budget families obeying $\xi^{\cal J}\cdot \pi<1$, cyclical choices must occur with positive probabilities between budgets $1$ and $2$ and between budgets $2$ and $3$.
For example, $\pi$ is represented as a convex combination of behavioral types by letting $a^1=(1,0,0;1,0,0;1,0,0)^t$, $a^2=(0,1,0;0,1,0;0,1,0)^t$, $a^3=(0,0,1;0,0,1;0,0,1)^t$, $a^4=(0,1,0;1,0,0;0,1,0)^t$ and $a^5=(0,1,0;0,0,1;0,1,0)^t$, and putting weights as $\tau_{a^1}=1/10$, $\tau_{a^2}=2/10$, $\tau_{a^3}=1/10$, $\tau_{a^4}=3/10$ and $\tau_{a^5}=3/10$. Using Lemma \ref{lem_sarp}, it is straightforward that $a^1$, $a^2$ and $a^3$ are rationalizable, while $a^4$ and $a^5$ are not. To be more specific, $a^4$ contains cyclical choices between budgets $1$ and $2$, while $a^5$ has cyclical choices between budgets $2$ and $3$, both of which can be also checked through the product $\Xi a^m$ ($m=4,5$).
\end{example}

\begin{example}\label{ex4}
Consider an example taken from KS (Example 3.2), where $n=3$ and $J=3$. The budget lines are respectively determined by price vectors $p_1=(1/2,1/4,1/4)$, $p_2=(1/4, 1/2, 1/4)$, and $p_3=(1/4,1/4,1/2)$. Each budget line in this example has four patches, and there are $I=12$ patches in total. Figure \ref{3-3eg} describes the situation. There, the patches are indexed so that patches $B_{j4}$ for $j=1,2,3$ are undominated in $\{1, 2, 3\}$; $B_{14}$, $B_{13}$, $B_{24}$, and $B_{22}$ are undominated in $\{1, 2\}$; $B_{24}$, $B_{23}$, $B_{34}$, and $B_{32}$ are undominated in $\{2, 3\}$; $B_{34}$, $B_{33}$, $B_{14}$, and $B_{12}$ are undominated in $\{1,3\}$. Accordingly, we have a $(4 \times 12)$-matrix
\begin{align*}
\Xi =
\begin{pmatrix} 
0 & 0 & 0 & 1 & 0 & 0 & 0 & 1 & 0 & 0 & 0 & 1\\
0 & 0 & 1 & 1 & 0 & 1 & 0 & 1 & 0 & 0 & 0 & 0\\ 
0 & 0 & 0 & 0 & 0 & 0 & 1 & 1 & 0 & 1 & 0 & 1\\ 
0 & 1 & 0 & 1 & 0 & 0 & 0 & 0 & 0 & 0 & 1 & 1
\end{pmatrix},
\end{align*}
where the row vectors are $\xi^{\{1,2,3\}}$, $\xi^{\{1,2\}}$, $\xi^{\{2,3\}}$, and $\xi^{\{1,3\}}$ in the order from top to bottom. Now, consider the stochastic demand system $$\pi=\left(0,\frac{1}{2},\frac{1}{2},0;0,\frac{1}{2},\frac{1}{2},0;0,\frac{1}{2},\frac{1}{2},0\right)^t.$$ Then, it holds that
$$\Xi\pi=\left(0, 1, 1, 1\right)^t\not\geq \mathbbm{1},$$
and hence, by Theorem \ref{main}, this stochastic demand system $\pi$ is not rationalizable. On the other hand, by Proposition \ref{violation}, cyclical choices is inevitable only on the set of budgets $\{1,2,3\}$, and the condition for WARP-consistency by Hoderlein and Stoye (2014) is satisfied. Thus, this is an example of the stochastic demand system consistent with WARP, but not rationalizable. Indeed, this $\pi$ is obtained as the midpoint of two behavioral types $a^1=(0, 1, 0, 0; 0, 1, 0, 0; 0, 1, 0, 0)^t$ and $a^2=(0, 0, 1, 0; 0, 0, 1, 0; 0, 0, 1, 0)^t$, both of which obey WARP, but not SARP.
\end{example}

\begin{figure}[tbp]
\centering
\includegraphics[keepaspectratio, scale=0.55]{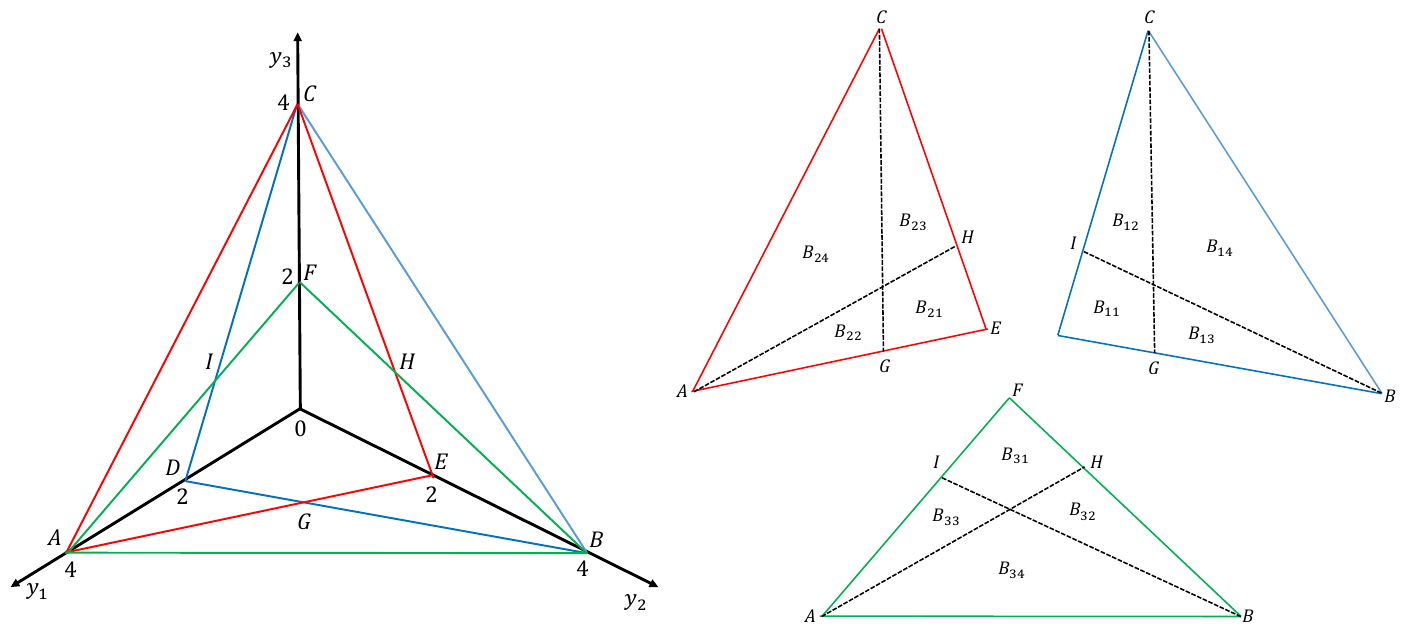}
\caption{Budget lines and patches in Example \ref{ex4}.}
\label{3-3eg}
\end{figure}

\vspace{-0.2in}
\subsection{Integrality of polytopes and the proof of Theorem \ref{main}}

From a mathematical viewpoint, as already referred to, Theorem \ref{main} says that the set of rationalizable stochastic demand systems is characterized as the intersection of finitely many half spaces in the form of $\{\pi:\xi^{\cal J}\cdot \pi\geq 1\}$. Put otherwise, the set of rationalizable demand systems is represented as\begin{align}
{\cal P}:=\bigcap_{{\cal J}\in \mathscr{J}}\{\pi:\xi^{\cal J}\cdot \pi\geq 1\}.
\end{align}
On the other hand, Lemma \ref{lem_sarp} says that\begin{align}
{\cal A}^*=\bigcap_{{\cal J}\in \mathscr{J}}\{a:\xi^{\cal J}\cdot a\geq 1\},
\end{align}
and hence, the set of rationalizable behavioral types is captured as the set of integral points of ${\cal P}$. Thus, the essential claim in Theorem \ref{main} is that ${\cal P}$ is an \emph{integral polytope} of which the set of vertices is equal to ${\cal A}^*$. (Strictly speaking, the hyperplanes corresponding to the nonnegativity and adding-up conditions should be also included, but they do not affect the following argument and are omitted.) Note that a polytope is said to be integral, if it is equal to the convex hull of its integer points. For example, the polytope in Figure \ref{inpol}(a) is not an integral polytope, while that in Figure \ref{inpol}(b) is an integral polytope. Note that, in Figure \ref{inpol}, lines represent hyperplanes, while dot points are regarded as integer points. In these polytopes, the sets of integral points are the same with each other. Thus, even if two sets of hyperplanes specify the same set of integer points, they may generate different polytopes. Rephrasing it in terms of the consumer theory, even if we obtain some matrix representation of rationalizable behavioral types, it  might not generate the set of rationalizable stochastic demand systems, which makes the claim of Theorem \ref{main} nontrivial.

\begin{figure}[tbp]
\centering
\includegraphics[keepaspectratio, scale=0.6]{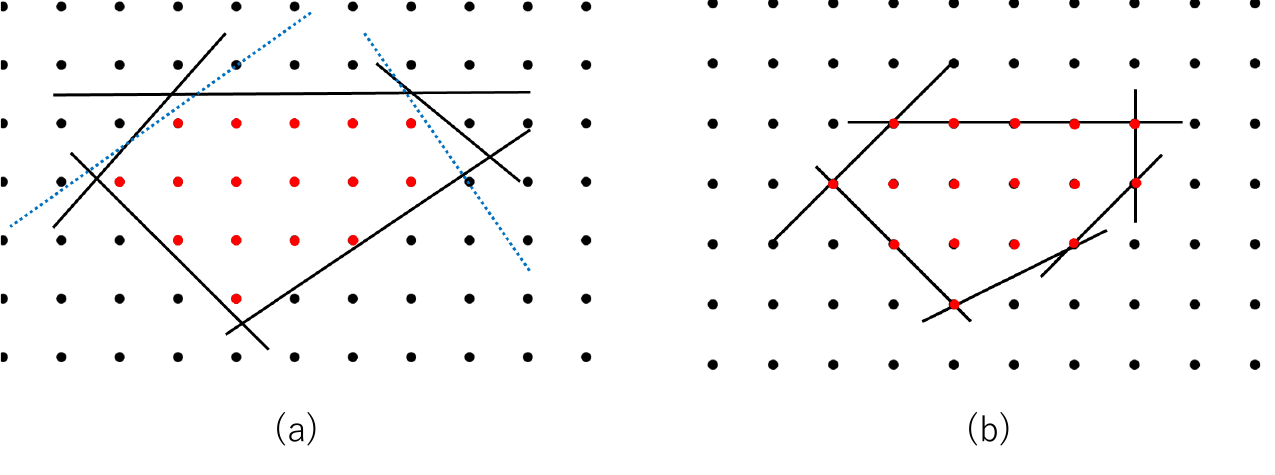}
\caption{Non-integral/Integral polytopes}
\label{inpol}
\end{figure}

To show that the hyperplanes generated by matrix $\Xi$ create a situation described as Figure \ref{inpol}(b) rather than that described as Figure \ref{inpol}(a), we use the notion of \emph{Chv\'{a}tal rank} explained below. In general, letting ${\cal Q}=\{q\in \mathbb{R}^L:  Wq\geq \theta\}$ and ${\cal Q}_{\cal I}=\mbox{conv.}\left({\cal Q}\cap \mathbb{Z}^L\right)$, ${\cal Q}$ is integral if ${\cal Q}={\cal Q}_{\cal I}$. The set ${\cal Q}_{\cal I}$ is referred to as the \emph{integral hull} of ${\cal Q}$. While the equality is not necessarily the case, ${\cal Q}_{\cal I}\subset {\cal Q}$ always holds, and, when $W$ is an $(M\times L)$-rational matrix, it is known that, the set $${\cal Q}^{(1)}:=\{q\in {\cal Q}: zW\mbox{ is integral for some }z\in [0,1]^M\Longrightarrow (zW)q\geq \lceil z\theta\rceil\}$$ is in between ${\cal Q}_{\cal I}$ and ${\cal Q}$. (Note that $\lceil\cdot \rceil$ stands for the ceiling function.)  That is, it holds that ${\cal Q}_{\cal I}\subset {\cal Q}^{(1)}\subset {\cal Q}$. This ${\cal Q}^{(1)}$ is referred to as \emph{Chv\'{a}tal closure} of ${\cal Q}$. Intuitively, Chv\'{a}tal closure adds some constraints to the original polytope ${\cal Q}$ to remove some non-integer vertices, such as dotted lines in Figure \ref{inpol}(a). Thus, ${\cal Q}^{(1)}$ is also a polytope, and defining ${\cal Q}^{(2)}$ as Chv\'{a}tal closure of ${\cal Q}^{(1)}$, it holds that ${\cal Q}_{\cal I}\subset {\cal Q}^{(2)}\subset {\cal Q}^{(1)}\subset {\cal Q}$. It is known that, repeating this procedure, there exists some $r<\infty$ such that ${\cal Q}^{(r)}={\cal Q}_{\cal I}$, with ${\cal Q}^{(0)}:={\cal Q}$. The minimum integer $r$ for which ${\cal Q}^{(r)}={\cal Q}_{\cal I}$ is called \emph{Chv\'{a}tal rank}, and hence, ${\cal Q}$ is integral if and only if its Chv\'{a}tal rank is equal to $0$. See, for example, Schrijver (1980) and Conforti, Cornu\'{e}jols, Zambelli (2014) for the detailed argument.\footnote{This means that for any polytope ${\cal Q}$, there is a finitely-many-step procedure to obtain its integral hull ${\cal Q}_{\cal I}$. In addition, since each step just adds finitely many linear inequalities, eventually one can obtain the matrix representation of ${\cal Q}_{\cal I}$. However, in general, this ``existence" is purely theoretical level, and it is typically hard to explicitly obtain ${\cal Q}_{\cal I}$, partly because the convergence is very slow and Chv\'{a}tal rank tends to be very large. Schrijver (1980) and Conforti et al. (2014) contains a fuller discussion concerning the upper bounds of Chv\'{a}tal rank of a given polytope. } \vspace{0.1in}

\noindent
{\sc Remark 2.} Another (and perhaps more prevalent) approach for ensuring the integrality of a polytope is to show that the matrix determining it is a \emph{totally unimodular matrix (TUM)}. A matrix is called a TUM, if the determinant of every square submatrix can only take value $0$ or $\pm 1$, which automatically implies that the matrix must consist of $0$ and $\pm 1$. While the matrices in our numerical examples are TUMs, it is not at all clear if it is generally the case. At least, a matrix $\Xi$ obtained in our procedure typically violates a well known sufficient condition for being a TUM, which prohibits a matrix from having more than two non-zero entries in each columns, in addition to another requirement concerning the property of row sums. (See Schrijver (1980) for the detail). Note also that, since we fix the RHS of the system of inequalities, the total unimodularity of $\Xi$ is a sufficient condition, while our condition concerning Chv\'{a}tal rank is a necessary and sufficient condition for ${\cal P}$ to be integral.

\vspace{-0.05in}
\subsubsection*{Proof of Theorem \ref{main}}
\vspace{-0.05in}

Given the above argument, to prove Theorem \ref{main}, it suffices to show that for every $\pi \in {\cal P}$ and $|{\mathscr J}|$-dimensional vector $u=(u_{\cal J})_{{\cal J}\in \mathscr{J}}$, it holds that\begin{align}\label{wts1}
u\Xi\mbox{ is integral }\Longrightarrow (u\Xi)\pi\geq \lceil u\mathbbm{1}\rceil,
\end{align}
which implies that ${\cal P}={\cal P}^{(1)}$, and hence ${\cal P}={\cal P}_{\cal I}$. Since $\lceil u\mathbbm{1}\rceil=\lceil \sum_{{\cal J}\in {\mathscr J}}u_{\cal J}\rceil$ holds in (\ref{wts1}), it suffices to show that there exists a partition of $\mathscr J$ such that the sum of $u_{\cal J}$'s on each component is an integer. (Then the total sum of $u$ is the sum of finitely many integers.) For this purpose, we use a profile of patches $(B_{11}, B_{21},...,B_{J1})$ constructed by adjustments of indices such that for $k=1,2,...,J-1$, (i) $B_{k1}$ is taken from $B_k$, and (ii) $B_{k1}$ is undominated in ${\cal J}_k:=\{k,k+1,...,J\}$, while it is dominated in any ${\cal J}_{k'}$ with $k'<k$.  

Such a profile always exists as long as $p_{j'}\neq p_{j''}$ for all $j'\neq j''$. For example, suppose that there is some commodity for which all $p_j$'s take different values from each other. With no loss of generality, we may consider it as commodity $1$ and sort price vectors so that  $p_{11}<p_{21}<\dots <p_{J1}$. Then, the patch on $B_1$ containing the vector $(1/p_{11},0,...,0)$ is clearly undominated in ${\cal J}_1$, and hence it works as $B_{11}$. The patch on $B_2$ containing the vector $(1/p_{21},0,...,0)$ would work as $B_{21}$, since it is dominated by $B_{1}$ through $B_{11}$, but not by any other budgets. The rest of $B_{k1}$ could be defined in a similar vein. Even if there is no such commodity (as in Example \ref{ex4}), one can find some vector $d\in \mathbb{R}^n_{+}\setminus \{0\}$ on the unit sphere such that price vectors are sorted as $p_1\cdot d<p_2\cdot d<\dots <p_J\cdot d$ with suitable adjustment of indices, because $J<\infty$. Then, a profile of patches $(B_{11}, B_{21},...,B_{J1})$ can be constructed in a similar way to the preceding case. That is, $B_{k1}$ is defined as the patch on $B_k$ containing the consumption vector $d/(p_k\cdot d)$.\footnote{Thus, one can regard the preceding case as the special case where $d=(1,0,...,0)$ works.}

Once we have constructed a profile of patches $(B_{11}, B_{21},...,B_{J1})$ as above, letting ${\mathscr J}_1=\{{\cal J}\in \mathscr{J}: {\cal J}\ni 1\}$, $B_{11}$ is undominated in every ${\cal J}\in {\mathscr J}_1$. Recalling the definition of $\xi^{\cal J}_{11}={\bf 1}(B_{11}\mbox{ is undominated in } {\cal J})$, this in turn implies that\begin{align}\label{p1eq}
(u\Xi)_{11}=\sum_{{\cal J}\in {\mathscr{J}_1}}u_{\cal J},
\end{align}
where the subscript in the LHS indicates the coordinate corresponding to the patch $B_{11}$. Since the vector $u\Xi$ itself is assumed to be integral, the above sum is also integral.  

By the construction of the profile $(B_{11}, B_{21},...,B_{J1})$, $B_{21}$ is undominated in ${\cal J}$, if and only if ${\cal J}\in {\mathscr J}_2:=\{{\cal J}\in {\mathscr J}:{\cal J}\ni 2\mbox{ and }{\cal J}\not\ni 1\}$. Thus, it holds that\begin{align}
(u\Xi)_{21}=\sum_{{\cal J}\in {\mathscr J}_2}u_{\cal J},
\end{align}
which is also integral. In addition, it holds that ${\mathscr J}_1\cap {\mathscr J}_2=\emptyset$. Similarly, also for $k\geq 3$, the patch $B_{k1}$ is undominated in ${\cal J}$ if and only if ${\cal J}\in {\mathscr J}_k:=\{{\cal J}\in \mathscr{J}: {\cal J}\ni k\mbox{ and } {\cal J}\not\ni j\mbox{ for }j\le k-1\}$. Then, it holds that\begin{align}
(u\Xi)_{k1}=\sum_{{\cal J}\in {\mathscr J}_k}u_{\cal J},
\end{align}
which is integral. Moreover, ${\mathscr J}_k\cap {\mathscr J}_{k-1}=\emptyset$ for all $k$, and hence ${\mathscr J}_1, {\mathscr J}_2,...,{\mathscr J}_k$ are mutually exclusive. Repeating this process up to $k=J-1$, we obtain ${\mathscr J}_1, {\mathscr J}_2,...,{\mathscr J}_{J-1}$ that are mutually exclusive and $\bigcup^{J-1}_{k=1}{\mathscr J}_k=\mathscr{J}$. (Recall that ${\mathscr J}$ is the family of budgets
with at least two elements.) Thus, ${\mathscr J}_1, {\mathscr J}_2,...,{\mathscr J}_{J-1}$ forms a partition of $\mathscr{J}$, and the sum of $u_{\cal J}$'s is integral on each ${\mathscr J}_k$ as desired. \qed
\vspace{-0.1in}
\section{Duality and its implication}
\vspace{-0.1in}
\subsection{The maximal weight on rational types}

While Theorem \ref{main} tests the rationality of a given stochastic demand system through the condition $\Xi\pi\geq \mathbbm{1}$, in fact, the value of $\Xi\pi$ also contains some information concerning the ``degree" of (ir)rationality. To be more specific, we claim that when $\pi$ is not rationalizable, $\min_{{\cal J}\in \mathscr{J}}\xi^{\cal J}\cdot \pi\in [0,1)$ is equal to the maximal possible weight on rational behavioral types to explain $\pi$. To include the case of rationalizable stochastic demand systems, we introduce  $\xi^{\overline{\cal J}}=(1/J,...,1/J;...;1/J,...,1/J)\in [0,1]^I$ and $\overline{\mathscr{J}}=\mathscr{J}\cup \{\overline{\cal J}\}$. Note that $\xi^{\overline{\cal J}}\cdot \pi=1$ holds for any stochastic demand system $\pi$. Using this extended set of indices $\overline{\mathscr J}$, we have the following.

\begin{theorem}\label{dual}
For a given stochastic demand system $\pi$, it holds that\begin{align}\label{thm2}
\min_{{\cal J}\in \overline{\mathscr{J}}}\xi^{\cal J}\cdot \pi=\max_{\tau\geq 0:A\tau =\pi}\sum_{a\in {\cal A}^*}\tau_a.
\end{align}
\end{theorem}\vspace{0.1in}

This theorem shows the duality between the characterization by Theorem \ref{ks} and that by Theorem \ref{main}. To see this, let $c\in \{0,1\}^{|\cal A|}$ such that $c_a={\bf 1}(a\in {\cal A}^*)$ for every $a\in \cal A$, and consider the linear programming\begin{align}\label{primal}
\max\, c\cdot \tau\,\mbox{ subject to } \pi=A\tau\mbox{ and }\tau\geq 0.
\end{align}
Then, it is obvious that $\pi$ is rationalizable, if and only if the value of the above problem, say, $P(\pi)$ is equal to $1$. The dual of the problem (\ref{primal}) is formulated as\begin{align}\label{dual_p}
\min\, \xi\cdot \pi\,\mbox{ subject to }\xi A\geq c,
\end{align}
in which, every $\xi^{\cal J}$ with ${\cal J}\in {\mathscr{J}}$ is feasible by Lemma \ref{lem_sarp}, while the feasibility of $\xi^{\overline{\cal J}}$ is obvious. Letting $D(\pi)$ be the value of problem (\ref{dual_p}), the duality theorem implies that $P(\pi)=D(\pi)$, but Theorem \ref{dual} makes a much stronger claim that $D(\pi)$ can be in fact achieved by one of finitely many vectors $\{{\xi}^{\cal J}\}_{{\cal J}\in \overline{\mathscr J}}$. 

It is obvious that for each $a\in {\cal A}$, $P(a)={\bf 1}(a\in {\cal A}^*)$, and hence, by the duality theorem, $D(a)={\bf 1}(a\in {\cal A}^*)$ holds as well. Thus, Theorem \ref{dual} is obvious for behavioral types. For each ${\cal J}\in \overline{\mathscr J}$, let us consider a subset of ${\cal A}$ that shares $\xi^{\cal J}$ as a solution to the dual problem (\ref{dual_p}):\begin{align}
{\cal A}^{\cal J}=\{a\in {\cal A}:D(a)=\xi^{\cal J}\cdot a\}.
\end{align}
Note that a single behavioral type may be contained in multiple ${\cal A}^{\cal J}$'s. In addition, since $\xi^{\overline{\cal J}}\cdot a=1$ for any $a\in {\cal A}^*$, it holds that ${\cal A}^*={\cal A}^{\overline{\cal J}}$. 
In fact, for a given stochastic demand system $\pi$, a representation $\pi=\sum_{a\in {\cal A}}\tau_aa$ \emph{achieves} the maximal possible weight on rational types in the sense that $P(\pi)=\sum_{a\in {\cal A}}\tau_aP(a)=\sum_{a\in {\cal A}^*}\tau_a$, if and only if the support of $\tau=(\tau_a)_{a\in {\cal A}}$ is a subset of some ${\cal A}^{\cal J}$ (${\cal J}\in \overline{\mathscr{J}}$). This property (in particular, ``if" part) plays a central role in the proof of Theorem \ref{dual}, while it seems also of independent interest. 

\begin{proposition}\label{Converse_dual}
(a) If a representation $\pi=\sum_{a\in {\cal A}}\tau_aa$ achieves $P(\pi)=\sum_{a\in {\cal A}^*}\tau_a$ for some $\tau\in \Delta({\cal A})$, then there exists some ${\cal J}\in \overline{\mathscr J}$ for which $\tau_a>0\Longrightarrow a\in {\cal A}^{\cal J}$. (b) Conversely, if $\pi$ is represented as $\pi=\sum_{a\in {\cal A}}\tau_aa$ for some $\tau\in \Delta({\cal A})$, with $\tau_a>0\Longrightarrow a\in {\cal A}^{\cal J}$ for some single ${\cal J}\in \overline{\mathscr{J}}$, then such a representation of $\pi$ achieves $P(\pi)=\sum_{a\in {\cal A}^*}\tau_a$. 
\end{proposition}

Gathering together with Theorem \ref{dual}, if one wishes to represent $\pi$ putting weights on rational behavioral types as much as possible, then, it suffices to look at the set of types ${\cal A}^{\cal J}$ for which ${\cal J}$ achieving the LHS of (\ref{thm2}). On the other hand, if one has obtained a representation of $\pi$ only by using types in a certain ${\cal A}^{\cal J}$, then it already achieves the maximal possible weight on rational types, and hence, such a family of budget ${\cal J}$ achieves the LHS of (\ref{thm2}). That is, the set of stochastic demand systems is partitioned into subsimplices generated by $\{{\cal A}^{\cal J}\}_{{\cal J}\in \overline{\mathscr{J}}}$, each of which provides a representation of $\pi$ achieving $P(\pi)$, the maximal possible weight on rational behavioral types.

The preceding proposition also indicates what kind of mixture can ``improve" the fit to a rational choice model. It is not surprising that, even if a (deterministic) demand pattern of each person is not rational, a mixture of them across a population is more or less consistent with a RUM. As the simplest case, consider a mixture of demand behavior of two people, where both of their behavioral types violate SARP. Applying Proposition \ref{Converse_dual}, however, if (and only if) those behavioral types do not have any $\xi^{\cal J}$ as a common solution to (\ref{dual_p}), then any nontrivial mixture of their behavior can admit positive weights on some rational behavioral types. A generalization of this property is formally proved as Lemma \ref{lem_ind} in the proof of Proposition \ref{Converse_dual}. (See also Lemma \ref{dual_lem} in Appendix.) 

Lastly, we refer to the logical relationship across results in this paper. It is not difficult to see that the statement of Theorem \ref{dual} implies Theorem \ref{main}. Indeed, $\pi\in \cal P$ holds if and only if the RHS of (\ref{thm2}) is equal to $1$, and $\min_{{\cal J}\in \overline{\mathscr{J}}}\xi^{\cal J}\cdot \pi=1$ immediately implies that $\Xi\pi\geq \mathbbm{1}$. Nevertheless, the proof of Theorem \ref{dual} in fact depends on Proposition \ref{Converse_dual}(a), and the proof of the latter in turn depends on Theorem \ref{main}. (The dependence is found in the proofs of lemmas in Appendix.) Thus, in the sense that assuming one of them derives others, Theorem \ref{main}, Proposition \ref{Converse_dual}(a) and Theorem \ref{dual} are logically equivalent.
 
\vspace{-0.1in}
\subsection{Numerical examples}

We raise three numerical examples to see how Theorem \ref{dual} and Proposition \ref{Converse_dual} actually work. First, we revisit Example \ref{ex3} in the preceding section, where $\pi$ is not rationalizable. In this example, $D(\pi)$ is attained by $\xi^{\cal J}$ corresponding to ${\cal J}$ containing three budgets, when there are only two commodities. That is, even in the two-commodity setting, choice patterns over more than two budgets may have a certain revealed preference implication. To be more specific, although they do not affect the result of ``0-1" test like Theorem \ref{main}, it could have some information in deriving (a specific type of) degree of rationality. Second, we reconsider Example \ref{ex4}, where despite the consistency with WARP, one cannot put any positive weight on rationalizable behavioral types to explain $\pi$. Lastly, we consider the case where a mixture of irrational behavioral types can admit a positive weight on rational behavioral types, and even can be fully rationalizable.

\setcounter{example}{2}

\begin{example}{\sc (Cont.d.)} Revisit the setting of Example \ref{ex3} in the previous section. There, it holds that $$D(\pi)=\xi^{\{1,2,3\}}\cdot \pi=\frac{2}{5},$$
and hence, by Theorem \ref{dual}, the maximal possible weight on rationalizable behavioral type is equal to this value of $2/5$. Recall also that $\xi^{\{1,2\}}\cdot \pi=\xi^{\{2,3\}}=7/10$ and $\xi^{\{1,3\}}\cdot \pi>1$. Hence, it suffices to look at pairs of budgets to conclude that this $\pi$ is not rationalizable, but we cannot find the maximal possible weight on rationalizable types without considering the triple of budgets $\{1,2,3\}$. It is straightforward to check that the representation of $\pi$ stated in the previous section in fact attains it: amongst five behavioral types raised there, $a^1$, $a^2$ and $a^3$ are rationalizable and the total weight on them is equal to $2/5$. In addition, it can be also confirmed that all of these five types are contained in ${\cal A}^{\{1,2,3\}}$: it holds that $1=\xi^{\{1,2,3\}}\cdot a^1=\xi^{\{1,2,3\}}\cdot a^2=\xi^{\{1,2,3\}}\cdot a^3$ and that $0=\xi^{\{1,2,3\}}\cdot a^4=\xi^{\{1,2,3\}}\cdot a^5$. 

On the other hand, $\pi$ can be also represented as the convex combination of the behavioral types $\bar{a}^1=(1,0,0;0,1,0;0,0,1)^t$, $\bar{a}^2=a^2$, $\bar{a}^3=(0,0,1;0,0,1;1,0,0)^t$, $\bar{a}^4=a^4$, and $\bar{a}^5=a^5$ with $\tau_1=1/10$, $\tau_2=1/10$, $\tau_3=1/10$, $\tau_4=4/10$, and $\tau_5=3/10$. Amongst these behavioral types, only $\bar{a}^1$ and $\bar{a}^2$ are rationalizable, and hence, in this representation, the total weight on rationalizable types is equal to $1/5<D(\pi)$. By Proposition \ref{Converse_dual}, this is caused by the fact that these types are not simultaneously contained in the same ${\cal A}^{\cal J}$ for any ${\cal J}\in \overline{\mathscr{J}}$. Indeed, $\bar{a}^1$ is only contained in ${\cal A}^{\overline{\cal J}}$, while $\bar{a}^3$, $\bar{a}^4$ and $\bar{a}^5$ are not contained in it.  
\end{example} 

\setcounter{example}{3}

\begin{example}
{\sc (Cont.d.)} Revisit the setting of Example \ref{ex4} in the previous section, wherein it holds that
$D(\pi)=\xi^{\{1, 2, 3\}} \cdot \pi=0$.
Accordingly, $\pi$ is irrational and the maximal possible weight on rationalizable behavioral types equals to zero, when it is consistent with WARP. In particular, the representation described in the previous section, $\pi$ as the midpoint of $a^1=(0, 1, 0, 0; 0, 1, 0, 0; 0, 1, 0, 0)^t$ and $a^2=(0, 0, 1, 0; 0, 0, 1, 0; 0, 0, 1, 0)^t$, attains $D(\pi)$. This follows from Proposition \ref{Converse_dual}, given the fact that $\xi^{\{1, 2, 3\}} \cdot a^1=\xi^{\{1, 2, 3\}} \cdot a^2=0$, which means that both $a^1$ and $a^2$ are contained in ${\cal A}^{\{1, 2, 3\}}$ (as can be also directly confirmed).  
\end{example}

\begin{example}\label{ex5}
Consider the budget lines in Figure \ref{fig_maxpatch}, and the following pair of behavioral types: $a^1=(1,0,0;0,0,1;1,0,0)^t$ and $a^2=(0,0,1;1,0,0;0,0,1)^t$. It is easy to check that both of them are not rationalizable, and that they do not share any $\xi^{\cal J}$ as a common solution to (\ref{dual_p}). Then, by Proposition \ref{Converse_dual}, any nontrivial mixture of them can admit a positive weight on rationalizable behavioral types. For example, letting $$\pi=\left(\frac{1}{2},0,\frac{1}{2};\frac{1}{2},0,\frac{1}{2};\frac{1}{2},0,\frac{1}{2} \right)^t,$$ this is even rationalizable, as $\Xi\pi\geq \mathbbm{1}$ is easily confirmed. In particular, $P(\pi)=1$ is achieved by representing $\pi$ as the midpoint of the following (rationalizable) behavioral types: ${\bar a}^1=(1,0,0;1,0,0;1,0,0)^t$ and ${\bar a}^2=(0,0,1;0,0,1;0,0,1)^t$.
\end{example}

\vspace{-0.1in}
\subsection{The proofs of results in this section}

\subsubsection*{Proof of Theorem \ref{dual} (based on Proposition \ref{Converse_dual})}

Given that $P(\pi)=D(\pi)$ holds by the duality theorem, the statement is essentially equivalent to $D(\pi)=\min_{{\cal J}\in \overline{\mathscr{J}}}\xi^{\cal J}\cdot \pi$. Since every $\xi^{\cal J}$ is feasible in the problem (\ref{dual_p}), it suffices to show the existence of some ${\cal J}'\in \overline{\mathscr{J}}$ for which $D(\pi)=\xi^{{\cal J}'}\cdot \pi$. This is trivial if $\pi$ is rationalizable, since $\pi\in \mbox{conv.}({\cal A}^{\overline{\cal J}})$ has to hold. Even if $\pi$ is not rationalizable, the claim can be easily proved, once we establish Proposition \ref{Converse_dual}(a). To see this, suppose that $\pi$ is not rationalizable and that $P(\pi)=\sum_{a\in {\cal A}^*}\tau_a$ for some $\tau\in \Delta(\cal A)$. Applying Proposition \ref{Converse_dual}(a), it holds that $\tau_a>0\Longrightarrow a\in {\cal A}^{{\cal J}'}$ for some ${\cal J}'\in \mathscr{J}$, which also implies that $\tau_a>0\Longrightarrow\xi^{{\cal J}'}\cdot a=D(a)$.\footnote{Note that ${\cal J}'$ must be found in ${\mathscr J}$, rather than $\overline{\mathscr{J}}$, since $\pi\notin {\cal P}$ (i.e. $\pi$ is not rationalizable) is assumed.} This leads to $$\xi^{{\cal J}'}\cdot \pi=\sum_{a\in {\cal A}^{{\cal J}'}}\tau_a D(a)=\sum_{a\in {\cal A}^{{\cal J}'}\cap {\cal A}^*}\tau_a=P(\pi)=D(\pi),$$ 
which is what we have to show.  \qed

\vspace{-0.05in}
\subsubsection*{Proof of Proposition \ref{Converse_dual}}
\vspace{-0.05in}

We start from part (b) of the statement. Suppose that $\pi$ is represented as a convex combination $\pi=\sum_{a\in {\cal A}}\tau_aa$ for some $\tau\in \Delta({\cal A})$, with obeying $\tau_a>0\Longrightarrow a\in {\cal A}^{\cal J}$ for a common ${\cal J}\in \overline{\mathscr J}$. It holds that\begin{align}
\xi^{\cal J}\cdot \pi=\xi^{\cal J}\cdot \left(\sum_{a\in {\cal A}}\tau_aa \right)=\sum_{a\in {\cal A}}\tau_a(\xi^{\cal J}\cdot a)=\sum_{a\in {\cal A}}\tau_aD(a).
\end{align}
Moreover, since $D(a)={\bf 1}(a\in {\cal A}^*)$, it also holds that $\sum_{a\in {\cal A}}\tau_aD(a)=\sum_{a\in {\cal A}^*}\tau_a$, and hence, $\xi^{\cal J}\cdot \pi=\sum_{a\in {\cal A}^*}\tau_a$. Thus, if $\xi^{\cal J}\cdot \pi>D(\pi)$ were to hold, by the duality theorem, we would have $\sum_{a\in {\cal A}^*}\tau_aa>P(\pi)$. However, since $\pi=\sum_{a\in {\cal A}}\tau_aa$ is assumed, this contradicts the definition of $P(\pi)$. Hence, it must hold that $\xi^{\cal J}\cdot \pi=D(\pi)$, and the duality theorem implies that $P(\pi)=\sum_{a\in {\cal A}^*}\tau_aa$ as desired.

To prove the other direction (part (a)), the following lemma plays a key role. This is a generalization of the phenomenon observed in Example \ref{ex5}, but the formal proof, which is postponed to Appendix, is rather involved.

\begin{lemma}\label{lem_ind}
Fix $a_1,a_2,...,a_m\in {\cal A}$ ($m\geq 2$) for which there exists no ${\cal J}\in \overline{\mathscr J}$ such that $a_1,a_2,...,a_m\in {\cal A}^{\cal J}$. If a stochastic demand system $\pi$ is represented as $$\pi=\frac{1}{m}a_1+\frac{1}{m}a_2+\dots +\frac{1}{m}a_m,$$ then it holds that $$D(\pi)>\frac{1}{m}\sum^m_{k=1}D(a_k).$$
\end{lemma}\vspace{0.25in}

Admitting this lemma, the rest of the proof is as follows. Suppose that a given stochastic demand system $\pi$ is represented as a convex combination of $a_1,a_2,...,a_m$ for which there is no ${\cal J}\in \overline{\mathscr{J}}$ such that $a_1,a_2,...,a_m\in {\cal A}^{\cal J}$. Letting $\pi=\sum^m_{k=1}\tau_ka_k$ with $\tau_1\geq \tau_2\geq\dots \geq \tau_m\geq 0$ with $\sum^m_{k=1}\tau_k=1$, it holds that\begin{align}
\pi &= \tau_m(a_1+\dots +a_m)+(\tau_{m-1}-\tau_m)(a_1+\dots a_{m-1})+\dots+(\tau_1-\tau_2)a_1\\
&=m\tau_m\left(\frac{1}{m}\sum^m_{k=1}a_k\right)+(m-1)(\tau_{m-1}-\tau_m)\left(\frac{1}{m-1}\sum^{m-1}_{k=1}a_k\right)+\dots+(\tau_1-\tau_2)a_1\\
&=\sum^{m}_{l=1}(m-l+1)(\tau_{m-l+1}-\tau_{m-l+2})\left(\frac{1}{m-l+1}\sum^{m-l+1}_{k=1}a_k\right),
\end{align}
where we let $\tau_{m+1}=0$. For each $l=1,2,...,m$, $\left(\frac{1}{m-l+1}\sum^{m-l+1}_{k=1}a_k\right)$ is a stochastic demand system, and nonnegative numbers $(m-l+1)(\tau_{m-l+1}-\tau_{m-l+2})$ for $l=1,2,...,m$, add up to $1$. Indeed, it holds that\begin{align}
\sum^{m}_{l=1}(m-l+1)(\tau_{m-l+1}-\tau_{m-l+2})=\sum^m_{k=1}\tau_k=1.
\end{align}
Using this, we obtain that\begin{align}
D(\pi)&\geq \sum^{m}_{l=1}(m-l+1)(\tau_{m-l+1}-\tau_{m-l+2})D\left(\frac{1}{m-l+1}\sum^{m-l+1}_{k=1}a_k\right)\\
&>\sum^m_{k=1}\tau_kD(a_k)=\sum^m_{k=1}\tau_k{\bf 1}(a_k\in {\cal A}^*),
\end{align}
where the first inequality follows from the concavity of $D(\cdot)$, while the latter holds by Lemma \ref{lem_ind}.\footnote{The concavity of $D(\cdot)$ follows from the fact that $D(\cdot)=P(\cdot)$ by the duality theorem. The concavity of $P(\cdot)$ is rather obvious, given that it is the value of the maximization problem (\ref{primal}).} This means that $P(\pi)$, which is equal to $D(\pi)$, is not achieved by $\tau\in \Delta({\cal A})$ whose support is not restricted to some ${\cal A}^{\cal J}$. This completes the proof of Proposition \ref{Converse_dual}(a). \qed

\vspace{-0.1in}
\section{Conclusion}

In this paper, we have developed a new approach to nonparametric characterization for RUMs. A keystone of our analysis is the construction of a matrix capturing the structures of revealed preference relations across patches in each subfamily of budgets. Then, using this matrix, we provide a quantifier-free necessary and sufficient condition under which a given stochastic demand system is rationalizable by a RUM. In our characterization, the set of rationalizable demand systems is captured as an intersection of finitely many half spaces, which corresponds the dual of the ``vertex-based" characterization by KS. 

Our characterization is something beyond checking the consistency with RUMs in that, especially when a given demand system is not rationalizable, one can simultaneously obtain (i) subfamily of budgets in which cyclical choices occur with positive probabilities and (ii) the maximal possible weight on rational behavioral types in a population. The former could be potentially useful to explore causes of irrational choices by checking, for example, any common structure among families of budgets in which irrational choices are inevitable. On the other hand, the latter would suggest the possibility of constructing some ``non-binary" test or rationality indices for RUMs. Such a work could be related to the index of rationality by Apesteguia and Ballester (2015), which is based on stochastic choices, as well as other rationality indices for deterministic models including those explained in the textbook by Chambers and Echenique (2018). 

As in KS and other related papers, potentially, the results in this paper can also be applied to empirical analysis. To deal with samples of choices rather than a choice distribution in a population level, one needs to establish some procedure for the statistical implementation. In the framework of consumer theory, KS provides a statistical test using bootstrap, which in fact uses theoretical nature of the dual of their characterization without explicitly knowing it. Now, having an explicit formulation of it by Theorem \ref{main}, one may further develop the statistical procedure for testing the consistency with RUMs. For example, as mentioned in KS, we may appeal to some technology developed in the literature of \emph{generalized moment selections (GMS)} such as Andrews and Soares (2010), Bugni (2010) and Canay (2010).\footnote{Note that Hoderlein and Stoye (2014) has actually developed a statistical procedure for their WARP test based on techniques along this line.} Amongst others, a recent work by Cox and Shi (2023) provides a tractable procedure for testing for moment inequality models that does not depend on simulation and tuning parameter.

Lastly, the approach in this paper seems also applicable to other models along the line of KS, such as a price preference model by Deb et al. (2023) and even in a game theoretic framework dealt with in Lazzati et al. (2024). In these models, stochastic choices are captured as a mixture of model-consistent deterministic choices. As in the proof of Theorem \ref{main}, the notion of Chv\'{a}tal closure is useful to check if a characterization for deterministic choices directly extends to that for mixtures of them. If it does, then one could obtain a dual representation, possibly with some economic implications from it. Even if not, then, the constraints newly added by taking Chv\'{a}tal closure could suggest additional behavioral restrictions to be considered.

\section*{Appendix}\vspace{-0.1in}
\subsection*{Proof of Lemma \ref{lem_ind}}

The proof of Lemma \ref{lem_ind} needs the following auxiliary result.

\begin{lemma}\label{dual_lem}
Fix a pair of behavioral types $a_1,a_2\in {\cal A}$ for which there exists no ${\cal J}\in \overline{\mathscr J}$ such that $a_1,a_2\in {\cal A}^{\cal J}$. If a stochastic demand system $\pi$ is represented as $$\pi=\frac{1}{2}a_1+\frac{1}{2}a_2,$$ then it holds that $$D(\pi)>\frac{1}{2}D(a_1)+\frac{1}{2}D(a_2).$$
\end{lemma}\vspace{0.02in}

\begin{proof}
The condition in the statement implies that at least one of $a_1$ and $a_2$ is outside of ${\cal A}^*$, since $a\in {\cal A}^*\iff \xi^{\overline{\cal J}}\cdot a=1=D(a)$. With no loss of generality, we may assume that $a_2\notin {\cal A}^*$. 

Suppose that $a_1\in {\cal A}^*$. Since $\pi=\frac{1}{2}a_1+\frac{1}{2}a_2$, it holds that for each ${\cal J}\in \mathscr{J}$, $$\xi^{\cal J}\cdot \pi=\frac{1}{2}\xi^{\cal J}\cdot a_1+\frac{1}{2}\xi^{\cal J}\cdot a_2\geq 1,$$ where the inequality holds, because $\xi^{\cal J}\cdot a_1\geq 1$ by Lemma \ref{lem_sarp}, $\xi^{\cal J}\cdot a_1=1\Longrightarrow \xi^{\cal J}\cdot a_2\geq 1$ and $\xi^{\cal J}\cdot a_2=0\Longrightarrow \xi^{\cal J}\cdot a_1\geq 2$.\footnote{The latter two properties are due to the assumption that there is no ${\cal  J}$ for which $a_1,a_2\in {\cal A}^{\cal J}$. For example, if $\xi^{\cal J}\geq a_2=0$, then, $a_2\in {\cal A}^{\cal J}$, and hence $a_1\notin {\cal A}^{\cal J}$ must hold. Since $a_1\in {\cal A}^*$ is assumed, the latter implies that $\xi^{\cal J}\cdot a_1>D(a_1)=1$. In addition, since $\xi^{\cal J}\cdot a_1$ must be an integer, it holds that $\xi^{\cal J}\cdot a_1\geq 2$. } Then, Theorem \ref{main} tells us that this $\pi$ is rationalizable, and hence, we obtain $$D(\pi)=P(\pi)=1>\frac{1}{2}=\frac{1}{2}D(a_1)+\frac{1}{2}D(a_2)$$ as desired.

Now, we turn to the case of $a_1\notin {\cal A}^*$. In fact, under the condition that there is no ${\cal J}\in \overline{\mathscr J}$ such that $a_1,a_2\in {\cal A}^{\cal J}$, we can find some $\hat{a}_1,\hat{a}_2\in {\cal A}$ such that $\hat{a}_1\in {\cal A}^*$ and $a_1+a_2=\hat{a}_1+\hat{a}_2$. In particular, the latter immediately implies that $$\pi=\frac{1}{2}\hat{a}_1+\frac{1}{2}\hat{a}_2,$$ which in turn implies that $$D(\pi)=P(\pi)\geq \frac{1}{2}>0=\frac{1}{2}D(a_1)+\frac{1}{2}D(a_2).$$ Thus, the remaining part of the proof is showing the existence of these $\hat{a}_1$ and $\hat{a}_2$, which we construct as follows. Note that the broad idea is that we ``exchange" some choices by $a_1$ and $a_2$ so that revealed preference cycles turning up in $a_1$ are resolved. Thus, unless otherwise specified in the procedure below, $\hat{a}_1(B_j)=a_1(B_j)$ and $\hat{a}_2(B_j)=a_2(B_j)$ for each budget $B_j$ (\emph{i.e.} choices remain the same when the exchange is not needed).

In what follows, let for $k=1,2,...,J-1$, ${\cal J}_k=\{k,k+1,...,J\}$ and $\mathscr{J}_k=\{{\cal J}\in \mathscr{J}: {\cal J}\ni k\mbox{ and }{\cal J}\not\ni k'\mbox{ for }k'\le k-1\}$. As we have shown in the proof of Theorem \ref{main}, $\{\mathscr{J}_k\}^{J-1}_{k=1}$ is a partition of $\{1,2,...,J\}$. Since there is no ${\cal J}$ such that $a_1,a_2\in {\cal A}^{\cal J}$, at least one of $\xi^{{\cal J}_k}\cdot a_1\geq 1$ and $\xi^{{\cal J}_k}\cdot a_2\geq 1$ holds for $k=1,2,...,J-1$. (Otherwise, since $a_1\notin {\cal A}^*$ and $a_2\notin {\cal A}^*$ are assumed in this part, $\xi^{{\cal J}_k}\cdot a_1=\xi^{{\cal J}_k}\cdot a_2=0$ implies that $a_1,a_2\in {\cal A}^{{\cal J}_k}$.) Note also that if $\xi^{{\cal J}_k}\cdot a_{m_k}\geq 1$ holds for ${m_k}\in \{1,2\}$, then there is some $B_{j_k}$ such that $a_{m_k}(B_{j_k})$ is undominated in ${\cal J}_k$. \begin{enumerate}
\item We start from $k=1$. If $a_1$ chooses a undominated patch $a_1(B_{j_1})$ in ${\cal J}_1$ for some integer $j_1\in {\cal J}_1$, then we just adjust the indices of budgets so that $j_1=1$. Thus, in this case, $\hat{a}_1(B_1)=a_1(B_1)$ and $\hat{a}_2(B_1)=a_2(B_1)$. On the other hand, if $a_1$ does not choose any undominated patch in ${\cal J}_1$, then $a_2(B_{j_1})$ is undominated in ${\cal J}_1$ for some $j_1$. (Otherwise, $a_1,a_2\in {\cal A}^{{\cal J}_1}$ holds.) In this case, in addition to adjusting indices so that $j_1=1$, we exchange the choices of $a_1$ and $a_2$ at $B_{1}$; that is, $\hat{a}_1(B_{1})=a_2(B_{1})$ and $\hat{a}_2(B_{1})=a_1(B_{1})$. This ensures that $\xi^{\cal J}\cdot \hat{a}_1\geq 1$ for all ${\cal J}\in \mathscr{J}_1$ at this stage.
\item Then, we move to $k=2$. Similar to the preceding case, at least one of $a_1$ and $a_2$ must attain a undominated patch in ${\cal J}_2$. Note that the definition of ${\cal J}_2$ must reflect the change of indices in the preceding step. If $a_1$ chooses a undominated patch $a_1(B_{j_2})$ in ${\cal J}_2$ for some integer $j_2\in {\cal J}_2$, then we adjust the indices of budgets as $j_2=2$, by which $a_1(B_2)$ is undominated in ${\cal J}_2$. In this case, $\hat{a}_1(B_2)=a_1(B_2)$ and $\hat{a}_2(B_2)=a_2(B_2)$. On the other hand, if  $a_1$ does not choose any undominated patch in ${\cal J}_2$, then $a_2(B_{j_2})$ must be a undominated patch in ${\cal J}_2$ for some integer $j_2\in {\cal J}_2$. In this case, we adjust the indices of budgets so that this $j_2$ equals to $2$ by which $a_2(B_2)$ is undominated, and then, exchange the choice made by $a_1$ with that by $a_2$ at $B_2$: $\hat{a}_1(B_2)=a_2(B_2)$ and $\hat{a}_2(B_2)=a_1(B_2)$. This ensures that $\xi^{{\cal J}}\cdot \hat{a}_1\geq 1$ for all ${\cal J}\in \mathscr{J}_2$. However, since we do not change the choices of $\hat{a}_1$ and $\hat{a}_2$ at $B_1$ from the preceding step, $\hat{a}_1(B_1)$ remains undominated in any ${\cal J}\in {\mathscr J}_1$. As a result, at this stage, it actually holds that $\xi^{{\cal J}}\cdot \hat{a}_1\geq 1$ for all ${\cal J}\in \mathscr{J}_1\cup\mathscr{J}_2$.
\item By induction, repeating the above procedure up to $k=J-1$, we obtain $\hat{a}_1$ for which $\xi^{\cal J}\cdot \hat{a}_1\geq 1$ for all ${\cal J}\in \bigcup^{J-1}_{k=1}\mathscr{J}_k$. (Similar to the preceding step, in defining ${\cal J}_k$, the changes of indices made before that point must be reflected.) However, since $\bigcup^{J-1}_{k=1}\mathscr{J}_k=\mathscr{J}$, Lemma \ref{lem_sarp} ensures that the resulting $\hat{a}_1$ is in fact an element of ${\cal A}^*$. (Recall that $\mathscr{J}$ is a family of ${\cal J}$'s with at least
two elements.) The claim that $\hat{a}_1+\hat{a}_2=a_1+a_2$ is obvious from the construction of $\hat{a}_1$ and $\hat{a}_2$; indeed, choices of patches made either by $\hat{a}_1$ or $\hat{a}_2$ are exactly those made either by $a_1$ or $a_2$.
\end{enumerate}\vspace{-0.3in}\end{proof}\vspace{0.1in}

We turn to the proof of Lemma \ref{lem_ind}. Recall that the hypothesis in the lemma implies that at least one of $a_1,a_2,...,a_m$ is outside of ${\cal A}^*$. Also, by Lemma \ref{dual_lem}, the statement is true for the case of $m=2$. Assuming that the lemma is proved up to $m-1$ ($m\geq 3$), by induction, we show that it is also true for the case of $m$.\vspace{0.1in}

\noindent
\emph{(Case 1.)} Suppose that one and only one of $a_1,a_2,...,a_m$ is outside of ${\cal A}^*$. Without loss of generality, we may let $a_1\notin {\cal A}^*$. Then, for any ${\cal J}$, it holds that $\xi^{\cal J}\cdot a_1\geq 0=D(a_1)$ and that $\xi^{\cal J}\cdot a_k\geq 1=D(a_k)$ for $k=2,...,m$. In addition, there is no ${\cal J}$ such that $\xi^{\cal J}\cdot a_k$ achieves $D(a_k)$ for all $k=1,2,...,m$, or equivalently, there is at least one $a_k$ for which $\xi^{\cal J}\cdot a_k>D(a_k)$.\footnote{Recall that, since both $\xi^{\cal J}$ and $a_k$ are binary vectors, $\xi^{\cal J}\cdot a_k$ can take a nonnegative integer value. Hence, $\xi^{\cal J}\cdot a_k>D(a_k)$ is equivalent to $\xi^{\cal J}\cdot a_k\geq D(a_k)+1$.} This immediately implies that for every ${\cal J}\in \overline{\mathscr{J}}$, \begin{align}
\xi^{\cal J}\cdot \pi=\frac{1}{m}\sum^m_{k=1}\xi^{\cal J}\cdot a_k\geq 1.
\end{align}
Theorem \ref{main} implies that $\pi$ is rationalizable, and hence, we have \begin{align}
D(\pi)=1>\frac{1}{m}\sum^m_{k=1}D(a_k)=\frac{m-1}{m}.
\end{align}

In the rest of this proof, suppose that $\{a_1,a_2,...,a_m\}$ contains at least two non-rationalizable behavioral types. \vspace{0.1in}

\noindent
\emph{(Case 2.)} Suppose that some $\{a_{\iota_1},...,a_{\iota_{m'}}\}\subset \{a_1,a_2,...,a_m\}$ ($2\le m'\le m-1$) does not have any ${\cal J}\in \overline{\mathscr{J}}$ with $a_{\iota_1},...,a_{\iota_{m'}}\in {\cal A}^{{\cal J}}$.\footnote{Note that, if there is only one irrational behavioral types, then the hypothesis here is not satisfied, since $\xi^{\overline{J}}\cdot a=D(a)$ for any $a\in {\cal A}^*$.}  We let $\{a_{\iota_1},...,a_{\iota_{m'}}\}=\{a_1,...,a_{m'}\}$, with no loss of generality, and consider $$\pi'=\frac{1}{m'}a_1+\frac{1}{m'}a_2+\dots +\frac{1}{m'}a_{m'}.$$ By appealing to the inductive assumption, it holds that\begin{align}
D(\pi')>\frac{1}{m'}\sum^{{m'}}_{k=1}D(a_k).
\end{align}

Combining the concavity of $D(\cdot)$ and the fact that $\pi$ can be written as $$\pi=\frac{m'}{m}\pi'+\frac{m-m'}{m}\left(\sum^m_{k=m'+1}a_k \right),$$ 
this immediately results in\begin{align}
D(\pi)\geq \frac{m'}{m}D(\pi')+\frac{m-m'}{m}\sum^m_{k=m'+1}D(a_k)>\frac{1}{m}\sum^m_{k=1}D(a_k),
\end{align}
as desired. \vspace{0.1in}

\noindent
\emph{(Case 3.)} Then, we deal with the case where every $\{a_{\iota_1},...,a_{\iota_{m'}}\}\subset \{a_1,a_2,...,a_m\}$ ($2\le m'\le m-1$) has some ${\cal J}\in \overline{\mathscr{J}}$ for which $a_{\iota_1},...,a_{\iota_{m'}}\in {\cal A}^{\cal J}$. Since at least two of $a_1,a_2,...,a_m$ are outside of ${\cal A}^*$, without loss of generality, we may assume that $a_1,a_2\notin {\cal A}^*$.  

We use the result of Case 2 by showing the following claim.\vspace{0.05in}

\begin{claim}\label{cl1}
There is a set of behavioral types $\{\hat{a}_1,\hat{a}_2,...,\hat{a}_m\}$ for which (i) $\sum^m_{k=1}\hat{a}_k=\sum^m_{k=1}a_k$, (ii) there is no ${\cal J}\in \overline{\mathscr{J}}$ such that  $\hat{a}_2,\hat{a}_3,...,\hat{a}_m\in {\cal A}^{\cal J}$, and (iii) $\hat{a}_k=a_k$ for $k\geq 3$.
\end{claim}

\begin{proof}
Letting ${\bf a}=\{a_1,a_2,...,a_m\}$, define $\mathscr{J}'({\bf a})\subset \overline{\mathscr{J}}$ such that\begin{align}
\mathscr{J}'({\bf a})=\left\{{\cal J}:\, \xi^{\cal J}\cdot a_1>D(a_1)\mbox{ and }\xi^{\cal J}\cdot a_k=D(a_k)\mbox{ for }k\geq 2 \right\}.
\end{align}
Note that, by the hypothesis of the current case, $\mathscr{J}'({\bf a})$ is nonempty. In what follows, we construct ${\bf a}'=\{{a}'_1,{a}'_2,...,{a}'_m\}$ such that $\sum^m_{k=1}{a}'_k=\sum^m_{k=1}a_k$ and $\mathscr{J}'({\bf a}')\subsetneq \mathscr{J}'({\bf a})$, by exchanging the choices of $a_1$ and $a_2$ with each other.

Let ${\cal J}'$ be a maximal element of $\mathscr{J}'({\bf a})$ in the sense of set inclusion. Since ${\cal J}'\in \mathscr{J}'({\bf a})$, there exists some $j'\in {\cal J}'$ for which $a_1(B_{j'})$ is undominated in ${\cal J}'$. By the construction of $\mathscr{J}'({\bf a})$, it is obvious that $a_2(B_{j'})$ is dominated in ${\cal J}'$. Then, define $a'_1$ and $a'_2$ such that $a'_1=a_1$ except that $a'_1(B_{j'})=a_2(B_{j'})$, and similarly, let $a'_2=a_2$ except that $a'_2(B_{j'})=a_1(B_{j'})$. Then, also letting $a'_k=a_k$ for $k\geq 3$, it immediately follows that $\sum^m_{k=1}{a}'_k=\sum^m_{k=1}a_k$.

By the assumption in the current case, there is some ${\cal J}$ for which $a_1,a_2\in {\cal A}^{\cal J}$. Since $a_1,a_2\notin {\cal A}^*$ is also assumed, it holds that $a'_1,a'_2\notin {\cal A}^*$. Indeed, no matter how exchanging the choices by $a_1$ and $a_2$, it is impossible to touch any undominated patch in ${\cal J}$, and hence, the resulting behavioral types cannot satisfy the necessary and sufficient condition for the rationalizability in Lemma \ref{lem_sarp}. It is also clear from our construction that $\xi^{{\cal J}'}\cdot a'_2\geq 1>0=D(a'_2)$, and hence, $a'_2\notin {\cal A}^{{\cal J}'}$. Thus, to show that $\mathscr{J}'({\bf a}')\subsetneq \mathscr{J}'({\bf a})$, it suffices to ensure that for any ${\cal J}''\notin \mathscr{J}'({\bf a})$, it cannot be an element of $\mathscr{J}'({\bf a}')$. 

In particular, we can concentrate on the case where ${\cal J}''$ is not an element of $\mathscr{J}'({\bf a})$ because of $\xi^{{\cal J}''}\cdot a_1=0$ or $\xi^{{\cal J}''}\cdot a_2\geq 1$. (If $\xi^{{\cal J}''}\cdot a_k>D(a_k)$ holds for some $k\geq 3$, then ${\cal J}''\notin \mathscr{J}'({\bf a}')$ is obvious.) In addition, if $\xi^{{\cal J}''}\cdot a_1=\xi^{{\cal J}''}\cdot a_2=0$, then $a_1,a_2,...,a_m\in {\cal A}^{{\cal J}''}$, which contradicts the hypothesis of the lemma itself. Thus, the only issue is the case of ${\cal J}''$ with $\xi^{{\cal J}''}\cdot a_2\geq 1$. For ${\cal J}''$ to be an element of $\mathscr{J}'({\bf a}')$, it has to hold that $\xi^{{\cal J}''}\cdot a'_2=0$. However, if it were to hold, then $j'\in {\cal J}''$ and it has to be the only one element for which $a_2(\cdot)$ is undominated in ${\cal J}''$. Thus, for all $j\in {\cal J}''\setminus \{j'\}$, the patch $a_2(B_j)$ is dominated by some budget within ${\cal J}''$. Gathering together with the fact that all chosen patches by $a_2$ on ${\cal J}'$ (including $a_2(B_{j'})$) are dominated in ${\cal J}'$, this implies that $\xi^{{\cal J}'\cup {\cal J}''}\cdot a_2=0$. If it also holds that $\xi^{{\cal J}'\cup {\cal J}''}\cdot a_1=0$, then $a_1,a_2,...,a_m\in {\cal A}^{{\cal J}'\cup {\cal J}''}$, contradicting the hypothesis in this lemma itself. However, if $\xi^{{\cal J}'\cup {\cal J}''}\cdot a_1>0$, then ${\cal J}'\cup {\cal J}''\in \mathscr{J}'({\bf a})$, which in turn contradicts the assumption that ${\cal J}'$ is maximal with respect to the set inclusion in $\mathscr{J}({\bf a})$. This ensures that $\mathscr{J}'({\bf a}')\subsetneq \mathscr{J}'({\bf a})$.

Then, setting ${\bf a}:={\bf a}'$, repeat this argument until $\mathscr{J}(\cdot)$ becomes the empty set. Once we get there, the resulting set of behavioral types can be adopted as $\{\hat{a}_1,\hat{a}_2,...,\hat{a}_m\}$ in the statement. Since $a'_k=a_k$ for $k\geq 3$ always holds in the procedure, the property (iii) is also satisfied.
\end{proof}

Having the preceding claim, the rest of the proof of Lemma \ref{lem_ind} is as follows. Letting $\{\hat{a}_1,\hat{a}_2,...,\hat{a}_m\}$ be a profile of behavioral types constructed in Claim \ref{cl1}, by the second requirement there, it holds that $$\pi=\frac{1}{m}\sum^m_{k=1}a_k=\frac{1}{m}\sum^m_{k=1}\hat{a}_k.$$ This in turn implies that\begin{align}
D(\pi)>\frac{1}{m}\sum^m_{k=1}D(\hat{a}_k)\geq \frac{1}{m}\sum^m_{k=1}D({a}_k),
\end{align}
where the first inequality holds, because the statement is true in Case 2, while the second inequality holds, because $D(a_1)=D(a_2)=0$ is assumed and $\hat{a}_k=a_k$ for $k\geq 3$. \qed

\end{document}